\documentclass[letterpaper, 10pt, conference]{ieeeconf}
\usepackage{amsmath,amssymb,paralist,subfigure,graphicx,color,stmaryrd,mathrsfs,url,algorithm2e,cite}
\usepackage{siunitx}
\usepackage{url}
\usepackage{algorithm2e}
\usepackage[table]{xcolor}
\usepackage{blkarray}
\usepackage{bbm}
\usepackage{bm}
\usepackage[makeroom]{cancel}
\usepackage{cases}
\usepackage{dsfont}
\usepackage{epsfig}
\usepackage[T1]{fontenc}
\usepackage{graphicx}
\usepackage{flushend}
\usepackage{textcomp}
\usepackage{xcolor}
\usepackage[UKenglish,english]{babel}
\usepackage{float}
\usepackage{flushend}
\usepackage{graphicx}
\usepackage{hyperref} 
\usepackage{mathtools}
\usepackage{mathrsfs}
\usepackage{setspace}
\usepackage{soul} 
\usepackage{subfigure}
\usepackage{stfloats}
\usepackage{tikz}
\usepackage{verbatim}
\usepackage{xspace}

\newtheorem{lem}{Lemma}
\newtheorem{ass}{Assumption}
\newtheorem{theorem}{Theorem}
\newtheorem{defn}{Definition}
\newtheorem{rem}{Remark}

\def\ve{\varepsilon}

\def\mb{\mathbf}

\def\mc{\mathcal}

\DeclareMathOperator*{\argmin}{argmin}
\DeclareMathOperator*{\argmax}{argmax}
\IEEEoverridecommandlockouts                
\begin{document}
\title{\Large \bf Fast-Convergent Anytime-Feasible Dynamics for Distributed Allocation of Resources Over  Switching Sparse Networks with Quantized Communication  Links}
\author{Mohammadreza Doostmohammadian, Alireza Aghasi, Mohammad Pirani, Ehsan Nekouei,\\ Usman A. Khan,~\IEEEmembership{Senior Member,~IEEE}, Themistoklis Charalambous,~\IEEEmembership{Senior Member,~IEEE}
\thanks{M. Doostmohammadian and T. Charalambous are with the School of Electrical Engineering at Aalto University,
Finland, \texttt{firstname.surname@aalto.fi}. 
M. Doostmohammadian is also with Semnan University, Iran, \texttt{doost@semnan.ac.ir}. 
T. Charalambous is also with the University of Cyprus, \texttt{surname.name@ucy.ac.cy}. A. Aghasi is with Georgia State University, GA, USA, email: \texttt{aaghasi@gsu.edu}. M. Pirani is with the University of Waterloo, Canada, \texttt{mpirani@uwaterloo.ca}. E. Nekouei is with  City University of Hong Kong, Hong-Kong, \texttt{enekouei@cityu.edu.hk}. U. A Khan is with  Tufts University, MA, USA, \texttt{khan@ece.tufts.edu}. }}
\maketitle
\thispagestyle{empty}

\begin{abstract}
	This paper proposes anytime feasible networked dynamics to solve resource allocation problems over time-varying multi-agent networks. The state of agents represents the assigned resources while their total (equal to demand) is constant. The idea is to optimally allocate the resources among the group of agents by minimizing the overall cost subject to fixed sum of resources. Each agent's information is local and restricted to its own state, cost function, and the ones from its immediate neighbors. 
	This work provides a fast convergent solution (compared to linear dynamics) while considering more-relaxed uniform network connectivity and (logarithmic) quantized communications among agents. The proposed dynamics reaches optimal solution over switching (sparsely-connected) undirected networks as far as their union over some bounded non-overlapping time-intervals has a spanning tree. Moreover, we prove anytime-feasibility of the solution, uniqueness, and convergence to the optimal value irrespective of the specific nonlinearity in the proposed dynamics. Such general proof analysis applies to many similar 1st-order allocation dynamics subject to \textit{strongly sign-preserving nonlinearities}, e.g., actuator saturation in generator coordination. Further, anytime feasibility (despite the nonlinearities) ensures that our solution satisfies the fixed-sum resources constraint at all times. 
	
	
\end{abstract}

\section{Introduction} \label{sec_intro}
Distributed optimization in machine learning, signal processing, and control literature, solves the following optimization of a global cost/objective as the sum of local functions:
\begin{equation} \label{eq_do}
\begin{aligned}
\displaystyle
\min_\mb{x}
& F(\mb{x}) = \sum_{i=1}^{n} f_i(\mb{x}) ~~ \mbox{subject to}~\mc{C}(\mb{x}) = 0 \\
\end{aligned}
\end{equation}
The centralized solution of \eqref{eq_do} works under the premise that all information is available and processed at a  central computing node.
However, in large-scale, every
node/agent only  has access to local information in its neighborhood, and distributed multi-agent algorithms  are needed to cooperatively perform local computations via local information. Such solutions find applications in multi-sensor target tracking \cite{acc13},
edge-computing and load balancing \cite{dechouniotis2020edge}, power allocation in cellular networks \cite{dra_vehicle},  and distributed support-vector-machine \cite{garg_cdc,dsvm}.
Different constraints and solutions are considered for problem~\eqref{eq_do}:  \textit{unconstrained}  \cite{xi2018linear,saadatniaki2020,nedic2014distributed} (for strongly-convex objectives), 
\textit{inequality constraint} \cite{chang2014distributed}, and 
\textit{consensus-constraint} \cite{garg_cdc,dsvm,feng2017finite} 
aiming  also to drive the agents to reach consensus along with optimization. 
In distributed resource allocation \cite{boyd2006optimal,doan2017scl, kar2012distributed,dominguez2012decentralized,doan_tac,shames2011accelerated,yi2016initialization,lin2020predefined, marden_dra} (also known as network resource allocation) the optimization  is constrained with constant summation of states, aiming to allocate a fixed amount of overall resources over a large-scale network. This finds applications in economic dispatch over power networks \cite{chen2016distributed,li2017distributed,cherukuri2015distributed,chen2018fixed,masoumzadeh2017impact,kar2012distributed,dominguez2012decentralized}, networked coverage control \cite{MiadCons,MSC09}, and  edge-computation offloading \cite{you2016energy}. 
Implementing parallel dynamics at agents via local information requires distributed algorithms,
some of which include: preliminary linear solution  \cite{boyd2006optimal}, quantized solution  via event-triggered  communications (fixed network)  \cite{li2021quantized}, accelerated linear solution via adding a momentum term (heavy-ball method) \cite{shames2011accelerated}, low communication rate protocol converging in quadratic time (via long-term connectivity requirement) \cite{doan2017scl}, game-theoretic approach \cite{marden_dra,nekouei2018impact}, initialization-free  \cite{yi2016initialization}, and Lagrangian-based solution \cite{kar2012distributed,dominguez2012decentralized,doan_tac}. For many existing solutions, as discussed in \cite{cherukuri2015distributed}, there is no guarantee for \textit{anytime} feasibility, i.e., the summation constraint is not necessarily feasible at all times but only at the final equilibrium state. Particularly, the presence of nonlinearities (e.g., saturation, quantization, or sign-based protocols) may affect solution uniqueness, all-time feasibility, and optimality.


\textbf{Contributions:} This work proposes a nonlinear
dynamics for network resource allocation. The main purposes for considering the nonlinearity are: fast convergence \cite{isj2020finite,polyakov2011nonlinear}, quantization \cite{rikos2020privacy,wei2018nonlinear,guo2013consensus,rikos2018distributed}, and saturation/clipping \cite{ wei2018nonlinear}, among others. Knowing that fixed-time dynamics reaches faster convergence than linear solutions \cite{polyakov2011nonlinear,garg_cdc}, we propose a continuous-time state-update for distributed resource allocation with fast convergence (as compared to linear solution) while considering logarithmic-quantized information exchange among agents. We consider uniform connectivity \cite{nedic2014distributed,saadatniaki2020} (instead of all-time connectivity), which only requires the \textit{union} network over some bounded non-overlapping time-intervals to include a spanning-tree. In contrast to unconstrained or consensus-constrained problems \cite{xi2018linear,saadatniaki2020,nedic2014distributed}, this work extends the solution to distributed resource allocation over sparsely-connected dynamic networks with quantized communications. Borrowing ideas from convex optimization and level-set methods \cite{bertsekas2003convex,aghasi2011parametric}, we prove uniqueness, feasibility, and optimal convergence under the proposed continuous-time dynamics. The convergence to this optimal value is proved via Lyapunov-type stability analysis.
Our convergence analysis, although given for a specific dynamics, can be easily extended to \textit{strongly sign-preserving nonlinear} $1$st-order solutions. Our main contributions are: (i) fast-convergence (compared to linear solutions) while considering quantized data transmissions, (ii) proving \textit{anytime} feasibility (e.g., versus Lagrangian-based solutions), and (iii) proving convergence \textit{for general strongly sign-preserving nonlinearities} (e.g., quantization and saturation)  over general uniformly-connected dynamic networks.      

\textbf{Paper organization:} Section~\ref{sec_prob} states the problem and preliminaries. Section~\ref{sec_solution} describes the proposed dynamics, while its convergence is proved in Section~\ref{sec_conv}. Section~\ref{sec_sim} provides simulations and Section~\ref{sec_conclusion} concludes the paper.

\section{Problem Formulation} \label{sec_prob}
\subsection{Problem Statement}
Distributed resource allocation problem is formulated as,
\begin{align} \label{eq_dra}
\min_\mb{x}
~~ & F(\mb{x}) = \sum_{i=1}^{n} f_i(\mb{x}_i)~~
\text{s.t.} ~ \sum_{i=1}^{n} \mb{x}_i = K
\end{align}
where $\mb{x}_i \in \mathbb{R}$ is the amount of resource allocated to agent $i$, $f_i:\mathbb{R} \rightarrow \mathbb{R} $ is a  convex function known by agent $i$ representing the cost as a function of resources $\mb{x}_i$. The network resource allocation problem~\eqref{eq_dra} aims to allocate a fixed quantity of total resources, $\sum_{i=1}^{n} \mb{x}_i = K$, among a group of agents communicating over an undirected graph $\mc{G}$, such that the total cost $F(\mb{x})$ is minimized.  In \eqref{eq_dra} the states should exactly meet the constraint $K$, e.g., in economic dispatch problem where the produced power is equal to the demand (anytime-feasibility). This differs from the \textit{inequality-constrained} power bid cost minimization problem in \cite{chang2014distributed} solved by primal-dual methods.
There might be box constraints $\underline{m}_i \leq \mb{x}_i \leq \overline{m}_i$ involved to bound the amount of resources. One can address these by adding  exact penalty functions \cite{bertsekas2003convex} to the objective as $f_i^\epsilon (\mb{x}_i) = f_i(\mb{x}_i) + \epsilon [\mb{x}_i - \overline{m}_i]^+ + \epsilon [\underline{m}_i - \mb{x}_i]^+$, with $[u]^+ = \max \{u,\mb{0}\}$. Recall that the summation of the strictly convex $f_i(\cdot)$ and convex penalty $[\cdot]^+$ is a strictly convex function. Further, the non-smooth $[u]^+$ can be replaced by its smooth equivalent $\frac{1}{\mu} \log (1+\exp (\mu u))$ as in \cite{dsvm} or quadratic penalty $([u]^+)^2$ \cite{nesterov1998introductory}.
Applications  include:
\begin{enumerate} [(i)]
	\item Economic dispatch \cite{chen2016distributed,li2017distributed,cherukuri2015distributed,masoumzadeh2017impact}: to allocate the electricity generation by facilities to minimize the cost while meeting the required load/demand constraints.
	\item Congestion-control and load-balancing \cite{you2016energy,dechouniotis2020edge}: to modulate traffics and data routing in telecommunication networks to gain fair allocations among the users.
	\item Coverage control \cite{MiadCons,MSC09}: the objective is to optimally allocate a group of networked robots/agents over a convex area in order to achieve maximum coverage.
\end{enumerate}	
\begin{rem}
Note the  difference of (constrained) problem \eqref{eq_dra}  with general (unconstrained) distributed optimization (as in \cite{xi2018linear,saadatniaki2020,nedic2014distributed}). 
Other than the constraint, for general distributed optimization the cost at all agents is the same function of $\mb{x}$, i.e., $F(\mb{x}) = \sum_{i=1}^{n} f_i(\mb{x})$, while in \eqref{eq_dra} the cost  at agent $i$ is only a function of $\mb{x}_i$, i.e., $F(\mb{x}) = \sum_{i=1}^{n} f_i(\mb{x}_i)$. This implies that in \eqref{eq_dra} agent $i$ only needs to know its own state $\mb{x}_i$ and not the other agents' states $\mb{x}_j,j\neq i$.
\end{rem}
Table \ref{tab_literature} compares solutions and different constraints in the literature.
\begin{table} [t]
	\centering
	 \caption{Overview of related literature on distributed optimization.}
	\begin{tabular}{|c|c|c|}
		\hline
		Reference& Solution & Constraint $\mc{C}(\mb{x})$\\
		\hline
		\cite{doan2017scl,boyd2006optimal,shames2011accelerated,chen2016distributed,kar2012distributed,dominguez2012decentralized} & $1$st-order  &  $\sum_{i=1}^{n} \mb{x}_i = K$\\
		\hline
		 \cite{yi2016initialization}& $2$nd-order  & $\sum_{i=1}^{n} \mb{x}_i = K$ \\
		\hline
		\cite{xi2018linear,saadatniaki2020,nedic2014distributed} & $2$nd-order & -- \\
		\hline
		\cite{chang2014distributed} &  $2$nd-order & $\sum_{i=1}^{n} g_i(\mb{x}_i) \preceq 0$ \\ \hline
		\cite{feng2017finite,garg_cdc} & $2$nd-order & $\mb{x}_1 =\dots=\mb{x}_n$ \\
		\hline
		\cite{dsvm,lin2016distributed} & $1$st-order & $\mb{x}_1 =\dots=\mb{x}_n$ \\
		\hline
	\end{tabular}
	\label{tab_literature}
\end{table}
Recall that the $1$st-order dynamics refers to the consensus-type protocols in the form,
${\dot{\mb{x}}_i = \sum_{j \in \mc{N}_i}f(\mb{x}_j-\mb{x}_i)}$,
while $2$nd-order dynamics are in the form, ${\dot{\mb{x}}_i = g(\mb{y}_i)}$, ${\dot{\mb{y}}_i = \sum_{j \in \mc{N}_i}h(\mb{y}_j-\mb{y}_i)}$.
\begin{rem}
	For synchronization and consensus,  the $1$st-order dynamics, compared to its similar $2$nd-order counterparts, is known to have  faster convergence \cite[page~32]{gupta_book}.
\end{rem}

\subsection{Preliminary Definitions and Lemmas}
The communication network of agents is modeled as a  sequence of (possibly) time-varying undirected graphs, denoted by $\mc{G}(t) = (\mc{V}, \mc{E}(t))$ with
$\mc{V} = \{1,\dots,n\}$. Two agents $i$ and $j$ can communicate/exchange
messages if and only if $(i, j),(j, i) \in \mc{E}(t)$. Define $\mc{N}_i(t)=\{j|(j,i)\in \mc{E}(t)\}$
as  neighbors of agent $i$ at
time $t$ and $n$ by $n$ matrix $W(t)$ as adjacency weight matrix of $\mc{G}(t)$, where $W_{ij} > 0$ if link $(i,j) \in \mc{E}(t)$ and $W_{ij} = 0$ if $(i,j) \notin \mc{E}(t)$.


\begin{defn} \label{def_tree}
	In the undirected graph $\mc{G}=(\mc{V},\mc{E})$, define a spanning tree as an undirected spanning subgraph in which any two vertices in $\mc{V}$ are connected by exactly one path. Such tree contains the minimum possible links in $\mc{E}$.
\end{defn}

\begin{ass} \label{ass_tree}
	There exists a sequence of non-overlapping bounded time-intervals, $[t_k, t_k + l_k]$, where the union network across each interval $\bigcup_{t=t_k}^{t_k+l_k}\mc{G}(t)$ has a spanning  tree, while $\mc{G}(t)$ might be sparsely-connected.
\end{ass}
Note that the above weak connectivity requirement ensures a path from any node $i$ to any node $j$ infinitely often. A situation where Assumption~\ref{ass_tree} does not hold is when 
the network is in the form of two separate graph components.
\begin{defn} \label{def_convex}
	(\cite{bertsekas2003convex}) Function $h(\mb{x}):\mathbb{R}^n \rightarrow \mathbb{R}$ is \textit{strictly} convex if $h(k\mb{x}_1 + (1-k)\mb{x}_2) < k h(\mb{x}_1) + (1-k)h(\mb{x}_2)$, $\forall \mb{x}_1, \mb{x}_2 \in \mathbb{R}^n,0< k <1$, or $\nabla^2 h(\mb{x})>0$ for twice differentiable $h(\mb{x})$; if $\nabla^2 h(\mb{x})>u>0$ it is \textit{strongly} convex.
\end{defn}

\begin{ass} \label{ass_conv}
	The functions $f_i(\mb{x}_i),i={1,\dots,n}$ in problem~\eqref{eq_dra} are strictly convex and differentiable.
\end{ass}

\begin{lem}[\cite{boyd2006optimal,shames2011accelerated}] \label{lem_equilibria}
Under the Assumption~\ref{ass_conv}, the resource allocation problem \eqref{eq_dra} has a unique optimal solution ${\mb{x}^*}$ for which
$ \nabla F({\mb{x}^*}) = {\psi}^*\underline{\mb{1}}_n$,
where $\nabla F({\mb{x}^*}) = (\frac{df_1}{d\mb{x}_1}(\mb{x}^*),\dots,\frac{df_n}{d\mb{x}_n}(\mb{x}^*))^\top$ denotes the gradient of $F$ at $\mb{x}^*$, and $\underline{\mb{1}}_n$ is the column vector of $1$'s.
\end{lem}
The optimality condition of Lemma~\ref{lem_equilibria} is simply the KKT condition,
with ${\psi}^*$ as the optimal Lagrange multiplier and $\underline{\mb{1}}_n$ as the gradient of the constraint \cite{bertsekas2003convex}.

\begin{defn}
	Define the feasible set of states as the affine space ${\mc{S}_K=\{\mb{x} \in \mathbb{R}^n|\sum_{i=1}^{n} \mb{x}_i={K} \}}$. 
\end{defn}

\begin{defn} \cite{bertsekas2003convex,aghasi2011parametric} Given $ h(\mb{x}): \mathbb{R}^n \rightarrow \mathbb{R} $,  define its level set  $L_\gamma(h)= \{\mb{x} \in \mathbb{R}^n|h(\mb{x})\leq \gamma \in \mathbb{R}\}$. For strictly convex $h$, $L_\gamma(h)$ is closed, compact, and strictly convex.
\end{defn}
The following lemma describes the solution of \eqref{eq_dra}.

\begin{lem} \label{lem_unique}
	Under Assumption~\ref{ass_conv}, there is a unique point $\mb{x}^*$ such that $\nabla F({\mb{x}^*}) = \psi^* \underline{\mb{1}}_n$ for every feasible set $\mc{S}_K$. In other words, given a feasible set $\mc{S}_K$ there is only one such point $\mb{x}^* \in \mc{S}_K$ for which $\frac{df_i}{d\mb{x}_i}(\mb{x}^*)=\frac{df_j}{d\mb{x}_j}(\mb{x}^*),~\forall i,j \in \{1,\dots,n\}$.
\end{lem}
\begin{proof}
	Following strict convexity of $F(\mb{x})$ (Assumption~\ref{ass_conv}), only one level sets, say $L_\gamma$, is adjacent to the constraint facet  $\mc{S}_K$, with touching only at one point $\mb{x}^*$
	for which $\nabla F({\mb{x}^*})$ is orthogonal to $\mc{S}_K$, i.e.,  $\frac{df_i}{d\mb{x}_i}(\mb{x}^*)=\frac{df_j}{d\mb{x}_j}(\mb{x}^*),\forall i,j$. By contradiction, assume $\mb{x}^{*1},\mb{x}^{*2}\in \mc{S}_K$ such that $\nabla F({\mb{x}^{*1}}) = \psi_1 \underline{\mb{1}}_n$ and $\nabla F({\mb{x}^{*2}}) = \psi_2\underline{\mb{1}}_n$, implying  that the strongly convex $L_\gamma$,  $\gamma = F_{\mb{x}^{*1}} = F_{\mb{x}^{*2}}$ is tangent to the \textit{affine} faucet $\mc{S}_K$ at two points $\mb{x}^{*1}$ and $\mb{x}^{*2}$, or two level sets $L_{F_{\mb{x}^{*1}}}$ and $L_{F_{\mb{x}^{*2}}}$ are adjacent to $\mc{S}_K$. Both cases contradict the strict convexity and closedness of the level sets. This proves the lemma by contradiction.
\end{proof}

\begin{lem} \label{lem_sum}
	Let $g_l : \mathbb{R} \rightarrow \mathbb{R}, l \in \{1,2\}$ be an odd mapping, matrix $W \in \mathbb{R}^{n \times n}$ be symmetric, and $\varphi \in \mathbb{R}^n$. Then,
	\begin{align} \nonumber
	\sum_{i =1}^n \varphi_i\sum_{j =1}^n W_{ij} g_2(g_1(\varphi_j) &-g_1(\varphi_i)) = \\
	- \frac{1}{2} \sum_{i,j=1}^n W_{ij}(\varphi_j-\varphi_i) &g_2(g_1(\varphi_j)-g_1(\varphi_i)).
	\end{align}
\end{lem}
\begin{proof}
    For every $i,j$, $W_{ij} = W_{ji}$, and $g_l(x) = -g_l(-x), l \in \{1,2\}$. Thus, we have,
 	\begin{align} \nonumber
	\varphi_i W_{ij} & g_2(g_1(\varphi_j) -g_1(\varphi_i)) + \varphi_j W_{ji}  g_2(g_1(\varphi_i) -g_1(\varphi_j))\\ \nonumber
	& = W_{ij}(\varphi_i-\varphi_j) g_2(g_1(\varphi_j)-g_1(\varphi_i)) \\
	& = W_{ji}(\varphi_j-\varphi_i) g_2(g_1(\varphi_i)-g_1(\varphi_j)).
	\end{align}
    and the proof follows. 
\end{proof}

We borrow results on nonsmooth analysis and set-value notions  from 
\cite{cortes2008discontinuous}, and skip the details due to space limitation. 
Define the \textit{generalized gradient} $\partial h:\mathbb{R}^n \rightarrow \mc{B}\{\mathbb{R}\}$ for a nonsmooth function $h:\mathbb{R}^n \rightarrow \mathbb{R}$ as,
\begin{align}
    \partial h(\mb{x})= co\{\lim_{i\rightarrow \infty} \nabla h(\mb{x}_i): \mb{x}_i \rightarrow \mb{x}, \mb{x}_i \notin \Omega_h \cup S \}
\end{align}
with \textit{co} denoting convex hull, $S \subset \mathbb{R}^n$ as a zero Lebesgue measure set, and $\Omega_h \in \mathbb{R}^n$ as the set of non-differentiable points in the domain of $f$. If $h$ is \textit{locally Lipschitz}, then $\partial h(x)$ is nonempty, compact, and convex, and  $\partial h:\mathbb{R}^n\rightarrow \mc{B}\{\mathbb{R}\}$, $\mb{x} \mapsto \partial h(\mb{x})$, is
upper semi-continuous and locally bounded. Then, the \textit{set-valued Lie-derivative} $\mc{L}_\mc{H} h : \mathbb{R}^n  \rightrightarrows  \mathbb{R}$ of function $h$ with respect to the dynamics $\dot{\mb{x}} \in \partial \mc{H}(\mb{x})$ is,
\begin{align}
     \mc{L}_\mc{H} f =  \{\eta \in \mathbb{R}| \exists \nu \in \mc{H}(\mb{x})~s.t.~ \zeta^\top \nu = \eta,~\forall \zeta \in \partial f(\mb{x})\}
\end{align}
We use this for nonsmooth Lyapunov analysis in Section~\ref{sec_conv}.

\section{The Proposed Solution} \label{sec_solution}
Consider linear dynamics to solve~\eqref{eq_dra} over undirected graphs (e.g., Laplacian-gradient model  in \cite{cherukuri2015distributed}),  
\begin{eqnarray}
\dot{\mb{x}}_i = -\eta_1\sum_{j \in \mc{N}_i} W_{ij} \left(\frac{d f_i}{d \mb{x}_i}-\frac{d f_j}{d \mb{x}_j}\right),
\label{eq_sol_lin}
\end{eqnarray}
where $\eta_1>0$ and we assume \textit{symmetric} weight matrix $W$ with $W_{ij}\geq0$. Note that for switching networks the RHS of~\eqref{eq_sol_lin} is discontinuous and the dynamics indeed represents a differential inclusion $\dot{\mb{x}} \in \partial \mc{H}(\mb{x})$. Throughout the paper for notation simplicity we use equality instead of "$\in$". 
By ideas from finite-time consensus \cite{ zuo2016distributed,taes2020finite}, to  reach faster convergence for $|{\frac{d f_i}{d \mb{x}_i}-\frac{d f_j}{d \mb{x}_j}|<1}$,  \eqref{eq_sol_lin} is modified  as,
\begin{align}
\dot{\mb{x}}_i = -\eta_1\sum_{j \in \mc{N}_i} W_{ij} \mbox{sgn}^{v}\left(\frac{d f_i}{d \mb{x}_i}-\frac{d f_j}{d \mb{x}_j}\right),
\label{eq_sol_sgn1}
\end{align}
where $0<v<1$ and 
$\mbox{sgn}^{v}(x)=x|x|^{v-1}$ which 
is non-Lipschitz at $0$  (for $0<v<1$). In general, non-Lipschitz dynamics  in the form ${\dot{\mb{x}}_i = -\eta_1\sum_{j \in \mc{N}_i} W_{ij} \mbox{sgn}^{v}( \mb{x}_i-\mb{x}_j)}$ is proved to reach finite-time convergence \cite{ parsegov2013fixed} and faster than linear case close to the equilibrium (because $|\mbox{sgn}^{v}(x)|>|x|$ for $|x|<1$). Further, consensus is in finite-time \cite{ parsegov2013fixed,taes2020finite}, however, with slow rate in regions far from the equilibrium (because  $|\mbox{sgn}^{v}(x)|<|x|$ for $|x|>1$). 
To overcome this, in fixed-time consensus protocols \cite{parsegov2013fixed},  typically a second term is added as $\mbox{sgn}^{v_2}(x)$ with $v_2>1$. Having $|\mbox{sgn}^{v_2}(x)|>|x|$ for $|x|>1$ implies faster convergence rate than the linear case for states far from the consensus equilibrium. Therefore, the combination of the two gives faster convergence tunable by parameters $v_1$, $v_2$. Therefore, this work proposes nonlinear $1$st-order dynamics to solve problem~\eqref{eq_dra} via,

\small
\begin{align}
\dot{\mb{x}}_i =  -  \sum_{j \in \mc{N}_i} W_{ij}\Bigl( \eta_1\mbox{sgn}^{v_1}(\frac{d f_i}{d \mb{x}_i}-\frac{d f_j}{d \mb{x}_j}) 
+  \eta_2\mbox{sgn}^{v_2}(\frac{d f_i}{d \mb{x}_i}-\frac{d f_j}{d \mb{x}_j})\Bigr)
\label{eq_solution0}
\end{align}\normalsize
with $0<v_1<1<v_2$, $0<\eta_2,\eta_1$, and $W_{ij}=W_{ji}\geq 0$, where $|\mbox{sgn}^{v_2}(x)+\mbox{sgn}^{v_1}(x)|>|x|$; therefore, for any $\frac{d f_i}{d \mb{x}_i},\frac{d f_j}{d \mb{x}_j}$,  $|\dot{x}_i|$ is greater under dynamics \eqref{eq_solution0} (compared to the linear case). Unlike  \cite{boyd2006optimal,shames2011accelerated} we do not assume $W$ to be (doubly) stochastic. Further, the convergence rate can be tuned by $v_1$, $v_2$. 
Next, consider quantized links\footnote{The proposed model~\eqref{eq_solution1} represents a system dynamics whose RHS changes in discrete-time, therefore, a hybrid state $\zeta=(\mb{x},\tau)$ with $\tau$ as the timer state and ${p:t \in \mathbb{R}_{\geq 0} \rightarrow P=\{1,2,\dots,\overline{p}\}}$ as the  index of the switching network~$\mc{G}_p$. 
Then, the switching flow map is $\mc{F}:{\dot{p} = 0,}{\dot{\tau} \in [0,\frac{1}{m\omega}]}$ along with the dynamics of $\dot{\mb{x}}$ in \eqref{eq_solution1} and domain set ${\zeta =(\mb{x},p,\tau) \in \mc{C} = R^{n} \times P \times [0,1]}$. Then, 
the \textit{jump map} is $\mc{J}:~\mb{x}^+ = \mb{x}, p^+ \in P, \tau^+=0$ over the jump domain set $\zeta\in \mc{D} = R^n \times P \times \{1\}$, implying that the hybrid system jumps to a new mode $p \in P$ whenever $\zeta \in \mc{D}$ with the interval length depending on the timer rate $\dot{\tau}$ for each mode $p$. For example, for the minimum-length switching time-interval $m\omega$, the rate is $\dot{\tau}=\frac{1}{m\omega }$. This simply means that after $m$ (and in general $m_p\geq m$)
periods $\omega$, the jump occurs as $\tau=1$ 
and $p$ switches to a new mode (a new topology $\mc{G}_p$), $\tau$ starts over, while the state $\mb{x}$ is continuous and unchanged. This is known as piece-wise constant jump mapping 
and satisfies the so-called \textit{Basic Assumptions} for stability, see  \cite{goebel2009hybrid,dsvm} and references therein.},  

\begin{align}  \nonumber
\dot{\mb{x}}_i =  &-  \sum_{j \in \mc{N}_i} W_{ij}(t,p)\Bigl( \eta_1\mbox{sgn}^{v_1}\big(q\big(\frac{d f_i}{d \mb{x}_i}\big)-q\big(\frac{d f_j}{d \mb{x}_j}\big)\big)\\
&+  \eta_2\mbox{sgn}^{v_2}\big(q\big(\frac{d f_i}{d \mb{x}_i}\big)-q\big(\frac{d f_j}{d \mb{x}_j}\big)\Bigr)
\label{eq_solution1}
\end{align}
where function $q_\rho(z)$ represents the data-quantization (of the shared information $\frac{d f_j}{d \mb{x}_j}$) as \cite{wei2018nonlinear,guo2013consensus}, 
\begin{align} 
	q_\rho(z) &= \mbox{sgn}(z) \exp(q_{u}(\log(|z|,\rho))), \label{eq_quan_log}
\end{align}
with $q_{u}(z) = \rho \left[ \frac{z}{\rho}\right]$ as the uniform quantizer ($[\cdot]$ denotes rounding to the nearest integer). The strongly sign-preserving odd function $q_\rho(\cdot)$ represents logarithmic data-quantization with level $\rho$, where  $(1-\frac{\rho}{2})z\leq q_\rho(z) \leq (1+\frac{\rho}{2})z$. 
\textit{For notation simplicity in the rest of the paper  we use  $\psi_i =: \frac{d f_i}{d \mb{x}_i}$ and  $W_{ij}$ instead of  $W_{ij}(p,t)$.} In \eqref{eq_solution1}, every agent knows its own state and objective along with information of its neighbors over $\mc{G}$. 
We consider \textit{periodic} communication, say every $\omega$ sec, with sufficiently small  $\omega$ (similar to \cite[Theorem~10]{KIA2015254}) with hybrid state as in \cite{dsvm,goebel2009hybrid}. Every agent $i$ shares $\psi_i = \frac{d f_i}{d \mb{x}_i}$ with $j \in \mc{N}_i(t)$ via (logarithmic) quantized channels, where the agent $j$ receives $q_\rho(\frac{d f_i}{d \mb{x}_i})$. 
Following the strict convexity of the level-sets of $F$ and \cite[Proposition S2]{cortes2008discontinuous}, for initialization  point $\mb{x}_0 \in \mc{S}_K$ and its level set $L_{F(\mb{x}_0)}$,
the solution in $ L_{F(\mb{x}_0)} \cap \mc{S}_K$ under the \textit{differential inclusion} \eqref{eq_solution1}  exists and is locally bounded, upper semi-continuous, with non-empty, compact, and convex values. We use this along with  $\mc{L}_\mc{H}$  referring  to  \textit{Lie-derivative} with respect to  $\dot{\mb{x}} = \partial \mc{H}(\mb{x})$ given by  differential inclusion \eqref{eq_solution1} in the rest of the paper.







\begin{lem} \label{lem_feasible}
	Consider a feasible $\mb{x}_0 \in \mc{S}_K$. For any symmetric weight matrix $W$, solution $x(t)$  keeps its feasibility (sum-preserving) under dynamics \eqref{eq_solution1} for $t>0$.
\end{lem}
\begin{proof}
	$\mb{x}_0 \in \mc{S}_K$ implies that $\sum_{i =1}^n {\mb{x}}_i(0)=K$. Then,
	\begin{align}
	\sum_{i =1}^n \dot{\mb{x}}_i 
	= &-  \sum_{i,j=1}^n  W_{ij}\Bigl( \eta_1\mbox{sgn}^{v_1}(q_\rho(\psi_i)-q_\rho(\psi_j)) \nonumber \\
	&+  \eta_2\mbox{sgn}^{v_2}(q_\rho(\psi_i)-q_\rho(\psi_j))\Bigr).
	\end{align}
	Note that  $\mbox{sgn}^{v}(q_\rho(\psi_j)-q_\rho(\psi_i)) = - \mbox{sgn}^{v}(q_\rho(\psi_i)-q_\rho(\psi_j))$ and $W_{ij}=W_{ji}$, therefore
	$\frac{d}{dt}\sum_{i =1}^n\mb{x}_i=\sum_{i =1}^n \dot{\mb{x}}_i = 0$.
\end{proof}

This lemma implies that initializing from a feasible $\mb{x}_0$, the solution under \eqref{eq_solution1} remains feasible \textit{at all times}. This \textit{anytime feasibility}, e.g., in generator coordination, outperforms \cite{kar2012distributed,dominguez2012decentralized,doan_tac} as the produced power does not deviate from the demand constraint, and  the feasibility set $\mc{S}_K$ is \textit{positively invariant } under dynamics~\eqref{eq_solution1}.
A feasible initialization could be simply $\mb{x}_i(0)=\frac{K}{n}$ at all agents (if within the box constraints) or via the algorithm in \cite{cherukuri2015distributed}. However, in many problems, the feasibility is initially satisfied; for example, in the economic dispatch, the initial energy production meets the demand (anytime-feasibility), and in the coverage control, the agents' coverage initially includes the entire convex area \cite{MiadCons} and the algorithm optimizes the production/coverage.
Unlike \cite{feng2017finite}, it is not necessary to constrain the initialization as $\sum_{i=1}^{n}\frac{\partial f_i}{\partial \mb{x}_i}(0)=0$.

\section{Convergence Analysis} \label{sec_conv}
In this section, we first characterize the unique equilibrium of dynamics~\eqref{eq_solution1} and then prove its stability.
\begin{theorem} \label{thm_tree}
	If Assumption~\ref{ass_tree} holds, the equilibrium point $\mb{x}^*$ with $\nabla F({\mb{x}^*}) = {\psi}^*\underline{\mb{1}}$ is the only equilibrium of \eqref{eq_solution1}.
\end{theorem}

\begin{proof}
	For $\mb{x}^*$ with $\nabla F({\mb{x}^*}) = {\psi}^*\underline{\mb{1}}$, we have $\dot{\mb{x}}^*_i=0$ which implies that $\mb{x}^*$ is the equilibrium of  \eqref{eq_solution1}. 
	By contradiction suppose there exists another equilibrium $\widehat{\mb{x}}$ under dynamics \eqref{eq_solution1} (i.e. $\dot{\widehat{\mb{x}}} = 0$) where $\frac{d f_i}{d \mb{x}_i}\neq\frac{d f_j}{d \mb{x}_j}$ for at least two nodes $i,j$. Let $\nabla F({\widehat{\mb{x}}}) = (\widehat{\psi}_1,\dots,\widehat{\psi}_n)^\top$ and consider nodes $a = \argmax_{\lambda\in \{1,\dots,n\}}  \widehat{\psi}_{\lambda }$ and $b = \argmin_{\lambda \in \{1,\dots,n\}}  \widehat{\psi}_\lambda$. Clearly, $\widehat{\psi}_a>\widehat{\psi}_b$. From Assumption~\ref{ass_tree} and Definition~\ref{def_tree}, there is (at least) one path between $a$ and $b$ in the union graph  $\bigcup_{t=t_k}^{t_k+l_k}\mc{G}(t)$, which includes two nodes $\alpha,\beta$ such that  $\widehat{\psi}_{\alpha}\geq \widehat{\psi}_{\mc{N}_{\alpha}}$ and $\widehat{\psi}_{\beta}\leq \widehat{\psi}_{\mc{N}_{\beta}}$ with strict inequality for (at least) one node in ${\mc{N}_{\alpha}}$  and $\mc{N}_{\beta}$. Both fixed-time  and quantization are sign-preserving odd functions, and thus, in a sub-domain of $[t_k,t_k+l_k]$, $\sum_{j \in \mc{N}_{\alpha}} \mbox{sgn}^{v}(q_\rho(\widehat{\psi}_{\alpha})-q_\rho(\widehat{\psi}_{j}))>0$ and  $\sum_{j \in \mc{N}_{\beta}} \mbox{sgn}^{v}(q_\rho(\widehat{\psi}_{\beta})-q_\rho(\widehat{\psi}_{j}))<0$ from dynamics \eqref{eq_solution1}. Therefore, $\dot{\widehat{\mb{x}}}_{\alpha} < 0$ and   $\dot{\widehat{\mb{x}}}_{\beta} > 0$ contradicting the equilibrium condition for~$\widehat{\mb{x}}$. This  proves the theorem.
\end{proof}

To illustrate more, assume that Assumption~\ref{ass_tree} does not hold; for example, the network has two separate connected components with nodes $1,\dots, m$ in component $\mc{G}_1$ and nodes $m+1,\dots, n$ in component $\mc{G}_2$. 
Then, for $\mb{x}^*$ as the equilibrium of~\eqref{eq_solution1}, $\frac{d f_i}{d \mb{x}_i}(\mb{x}^*_i)=\frac{d f_j}{d \mb{x}_j}(\mb{x}^*_j)=\psi^{*1}$ for $i,j \in \{1,\dots,m\}$ and $\frac{d f_i}{d \mb{x}_i}(\mb{x}^*_i)=\frac{d f_j}{d \mb{x}_j}(\mb{x}^*_j)=\psi^{*2}$ for $i,j \in \{m+1,\dots,n\}$ where in general $\psi^{*1} \neq \psi^{*2}$. Therefore, for such cases the equilibrium in general may differ from the optimal solution of the problem~\eqref{eq_dra} stated in Lemma~\ref{lem_unique}.

\begin{rem} \label{rem_unique}
	In the virtue of Lemma~\ref{lem_feasible}, the solution $\mb{x}(t)$ remains feasible under dynamics~\eqref{eq_solution1}, i.e., $\mb{x}(t) \in \mc{S}_K$ for  $t>0$. Then, following Lemma~\ref{lem_unique}, for initial value $\mb{x}_0 \in \mc{S}_K$  there is only one equilibrium $\mb{x}^*$ satisfying $\nabla F({\mb{x}^*}) = {\psi}^*\underline{\mb{1}}_n$.
\end{rem}
	
\begin{theorem} \label{thm_converg}
	Under Assumptions~\ref{ass_tree} and \ref{ass_conv} and starting from a feasible state  $\mb{x}_0 \in \mc{S}_K$, dynamics \eqref{eq_solution1} solves the resource allocation problem \eqref{eq_dra} (with feasibility at all times).
\end{theorem}

\begin{proof}	
Following Lemma~\ref{lem_equilibria}, for $\mb{x}^*$ as the optimal solution of problem~\eqref{eq_dra}, $\nabla F({\mb{x}^*}) = {\psi}^*\underline{\mb{1}}_n$.  
From Lemma~\ref{lem_feasible} $\mb{x}_0 \in \mc{S}_K$ implies solution feasibility  under  \eqref{eq_solution1}, and therefore, $\sum_{i=1}^n \mb{x}^*_i=K$. Let $F^*=F(x^*)$ and $\overline{F}(\mb{x})=F(\mb{x})-F^*$ denoting the residual with respect to optimal value. Consider this positive $\overline{F}(\mb{x})$ as the Lyapunov function with unique equilibrium $\overline{F}(\mb{x}^*)=0$. From \cite[Proposition~10]{cortes2008discontinuous},
for this nonsmooth, regular, and locally Lipschitz Lyapunov function $\overline{F}(\mb{x})$,  its derivative satisfies $\partial \overline{F}(\mb{x}(t))= \mc{L}_\mc{H} \overline{F}(\mb{x}(t))$ \cite[Proposition~10]{cortes2008discontinuous} with $\mc{H}$ referring to dynamics \eqref{eq_solution1}.
We show here that this non-negative Lyapunov function is monotonically non-increasing under dynamics \eqref{eq_solution1}. We have,

\begin{align}\nonumber
\partial \overline{F}(\mb{x}) 
=  &\sum_{i =1}^n \psi_i\Bigl(- \sum_{j \in \mc{N}_i} W_{ij}\Bigl(\eta_1\mbox{sgn}^{v_1}(q_\rho(\psi_i)-q_\rho(\psi_j)) \nonumber \\ &+  \eta_2\mbox{sgn}^{v_2}(q_\rho(\psi_i)-q_\rho(\psi_j))\Bigr)\Bigr). \nonumber
\end{align}
Then, in Lemma~\ref{lem_sum} set  $g_2(\cdot)$ as $\mbox{sgn}^{v}(\cdot)$  and  $g_1(\cdot)$ as $q_\rho(\cdot)$,
\begin{align} \label{eq_Frate}
\partial \overline{F}(\mb{x}) =  &-\frac{\eta_1}{2}\sum_{i,j =1}^n  W_{ij}|q_\rho(\psi_i)-q_\rho(\psi_j)|^{v_1+1} \nonumber \\ &- \frac{\eta_2}{2}\sum_{i,j =1}^n  W_{ij} |q_\rho(\psi_i)-q_\rho(\psi_j)|^{v_2+1}.
\end{align}
Therefore, $\partial \overline{F}(\mb{x}) \leq0$.
Following Lemma~\ref{lem_unique},  Remark~\ref{rem_unique}, and Theorem~\ref{thm_tree} the unique point $\mb{x}^*$ satisfying $\nabla F \in \mbox{span}\{\mb{1}_n\}$, is the unique equilibrium of  dynamics~\eqref{eq_solution1}, and, based on Lemma~\ref{lem_equilibria}, it is the optimal solution to the problem \eqref{eq_dra}. For any initial value $\mb{x}_0 \in \mc{S}_K$, the compact and closed (affine) solution set $\{L_{F(\mb{x}_0)} \cap \mc{S}_K\}$ is anytime feasible and positively invariant under \eqref{eq_solution1}. Thus, using LaSalle invariance principle for differential inclusions \cite[Theorem~2.1]{cherukuri2015distributed}, the solution converges to the largest invariant set $\mc{I}$ contained in $\{\mb{x} \in L_{F(\mb{x}_0)} \cap \mc{S}_K| 0 \in \mc{L}_\mc{H} \overline{F}(\mb{x}(t))\}$. Since $\mc{I} = \{\mb{x}^*\}$,    $\overline{F}\leq 0$, and $ \max \mc{L}_\mc{H} \overline{F}(\mb{x}(t)) < 0$ for all $\mb{x} \in \mc{S}_K \setminus \mc{I}$, dynamics \eqref{eq_solution1} globally asymptotically converges to $\mc{I} = \{\mb{x}^*\}$  \cite[Theorem~1]{cortes2008discontinuous}. This completes the proof.
\end{proof}
Recall that in Lemma~\ref{lem_sum} the oddness ensures anytime-feasibility.
The connectivity requirement in Assumption~\ref{ass_tree} gives the unique optimal state $\mb{x}^*$ (with $\nabla F({\mb{x}^*}) = {\psi}^*\underline{\mb{1}}$) in Theorem~\ref{thm_tree}, while Theorem~\ref{thm_converg} proves convergence to $\mb{x}^*$. 
\begin{rem}
The following gives an estimate of the convergence rate of Eq. \eqref{eq_Frate}. 
\begin{align} \nonumber
    |q_\rho(\psi_i)-q_\rho(\psi_j)|^{v_1+1} + &|q_\rho(\psi_i)-q_\rho(\psi_j)|^{v_2+1} \\ 
    &\geq |q_\rho(\psi_i)-q_\rho(\psi_j)|^2
\end{align}
where the RHS of the above represents the convergence rate of the linear (and linear quantized) protocols \cite{li2021quantized,boyd2006optimal}. Thus, the dynamics \eqref{eq_solution1} is faster than its linear counterparts. 
\end{rem}
\begin{rem}
The results of this paper can be extended to consider \textit{saturation} effects  \cite{ wei2018nonlinear}. For example, in case of actuator saturation one may substitute $\mbox{sgn}(x)$ or $\mbox{sat}_\kappa(\cdot)$ instead of $\mbox{sgn}^{v}(\cdot)$ in dynamics~\eqref{eq_solution1}, where $\mbox{sat}_\kappa(x) = x$ for $ -\kappa \leq x \leq \kappa$ and $\kappa \mbox{sgn}(x)$ otherwise.
Recall that the proofs of the given theorems and lemmas use only having non-zero derivative at zero, oddness, and sign-preserving property of $\mbox{sgn}^{v}(\cdot)$, which are also true for $\mbox{sat}_\kappa(\cdot)$ function. Therefore, the uniqueness, feasibility, and convergence results can be restated for general strongly sign-preserving nonlinearities on the agent's dynamics and communications. 
\end{rem}

\section{Numerical Simulations} \label{sec_sim}
For the simulations, 
smooth penalty $([u]^+)^2$ \cite{nesterov1998introductory} (with $\epsilon = 1$) for the box constraints is used  to satisfy Assumption~\ref{ass_conv}.


\subsection{A Comparison Study}
Consider the strictly-convex costs  as \cite{doan2017scl},
\begin{eqnarray} \label{eq_F1}
f_i(\mb{x}_i) = b_i(\mb{x}_i-a_i)^4,
\end{eqnarray}
with random coefficients $b_i \in (0,4]$, $a_i \in [-2,4]$, and box constraints $0\leq x_i\leq 5$. The random initial states satisfy $\sum_{i=1}^n \mb{x}_i(0) = K=20$ (as in Lemma~\ref{lem_feasible}). 
The multi-agent network is a cycle of $n=10$ nodes with random \textit{stochastic} link weights (this is required by \cite{boyd2006optimal,shames2011accelerated} and only for the sake of comparison). In Fig.~\ref{fig_compare}, the convergence of the dynamics \eqref{eq_solution1} is compared with linear \cite{boyd2006optimal}, accelerated linear ($\beta = 0.6$) \cite{shames2011accelerated}, quantized linear (with all-time triggered communications)  \cite{li2021quantized}, finite-time \cite{chen2016distributed}, and node-based fixed-time \cite{chen2018fixed} (with time-period $\omega = \num{2e-5}$, $\eta=1$, and $v_1=0.1$, $v_2=1.6$ in dynamics~\eqref{eq_solution1}).
\begin{figure}[]
	\raggedleft 
	\includegraphics[width=3.2in]{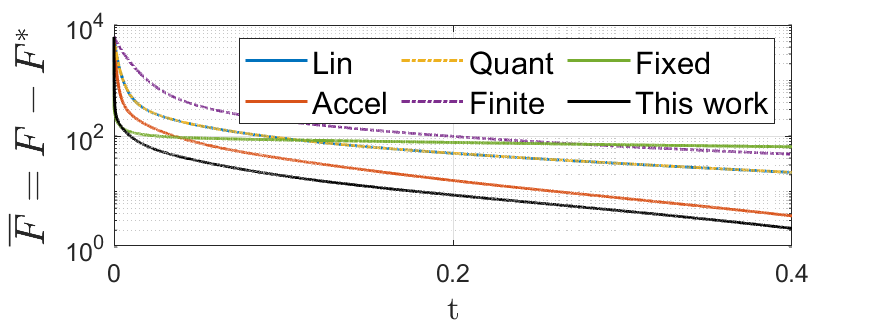}
	\caption{This figure compares the time-evolution of the residual under the proposed quantized dynamics (for quadratic cost) with some literature. } 
	\label{fig_compare} \vspace{-0.3cm}
\end{figure}


\subsection{Simulation over Weakly Connected Sparse Networks} 
We consider sparse networks of $n=100$ agents 
every $0.05$ sec switching between Scale-Free networks, 
none of which includes a spanning tree, while $\bigcup_{t=t_k}^{t_k+0.2}\mc{G}(t)$ is connected, i.e., $l_k=0.5$  ($100$ time-periods with $\omega = \num{5e-3}$) in Assumption~\ref{ass_tree}. The link weights are random and \textit{non-stochastic}.  Similar to \cite{boyd2006optimal}, consider strictly-convex objectives as,
\begin{eqnarray} \label{eq_F2}
f_i(\mb{x}_i) = \frac{1}{2}a_i(\mb{x}_i-c_i)^2 + \log(1+\exp(b_i(\mb{x}_i-d_i))),
\end{eqnarray}
with random coefficients $a_i \in (0,0.1]$, $b_i \in [-0.01,0.01]$, $c_i,d_i \in [-0.5,0.5]$ and box constraint $3 \leq \mb{x}_i \leq 7$.
The time-evolution of the states and absolute residual cost  $\overline{F}(\mb{x})$ is shown in Fig.~\ref{fig_sim}  with random initialization $\sum_{i=1}^{n} \mb{x}_i(0)= K = 500$, $\eta=1$, $\rho=\num{5e-4}$, and $v_1=0.3,~v_2=1.6$.  
\begin{figure}[]
	\raggedleft 
	\includegraphics[width=3.2in]{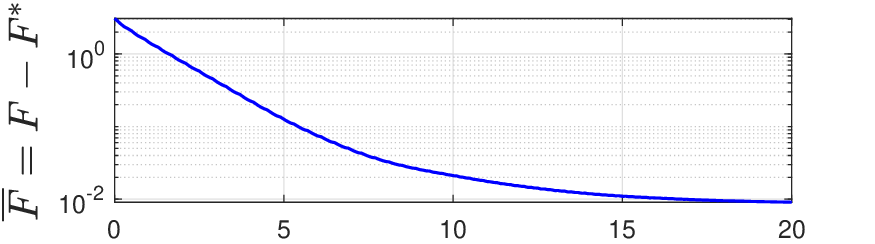}\vspace{0.1cm}
	\includegraphics[width=3.2in]{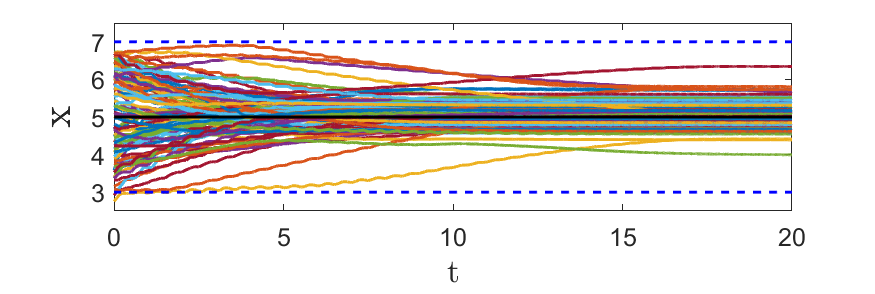}
	\caption{ Time-evolution of the residual (bottom) and the states (top) under dynamics \eqref{eq_solution1}. The box constraint and feasibility constraint (average of state values) are respectively shown by (dashed) blue and (solid) black  lines.  }
	\label{fig_sim} \vspace{-0.6cm}
\end{figure}

\section{Conclusion} \label{sec_conclusion}
This work provides a distributed nonlinear $1$st-order solution for resource allocation over dynamic undirected networks subject to (logarithmic) quantized data transmission, with convergence proved over sparse (uniformly-connected) networks.
The explicit discretization of \eqref{eq_solution1} (e.g., via Euler Forward method) can be used in real implementations assuming certain lower-bound on the sampling rate. For \textit{uniform} quantization with $ \frac{dq_{u}}{dz}|_{-0.5 \rho <z<0.5 \rho}=0$ (sign-preserving but \textit{not strongly}), one can prove convergence to $\ve$-neighborhood of $\mb{x}^*$, as a direction of our future research.



\bibliographystyle{IEEEbib}
\bibliography{bibliography}
\end{document}